\documentclass[11pt,a4paper]{article}
\usepackage[margin=1in]{geometry}
\usepackage[utf8]{inputenc} 
\usepackage[T1]{fontenc}
\usepackage[tbtags]{amsmath}
\allowdisplaybreaks
\usepackage{amssymb}
\usepackage{amsthm}
\usepackage{bm}
\usepackage[ttscale=.875]{libertine}
\usepackage[libertine,vvarbb]{newtxmath}
\usepackage[svgnames]{xcolor}
\usepackage{hyperref}
\hypersetup{colorlinks={true},linkcolor=[named]{DarkGreen},citecolor={DarkBlue}}
\usepackage[capitalise,nameinlink]{cleveref}
\usepackage{tikz}
\usetikzlibrary{math,patterns}

\usepackage{booktabs}
\usepackage[font=footnotesize]{caption}
\newtheorem{theorem}{Theorem}
\newtheorem{lemma}{Lemma}
\theoremstyle{definition}
\newtheorem{definition}{Definition}
\theoremstyle{remark}
\newtheorem{remark}{Remark}
\def \R{\mathbb R}
\newcommand{\sset}[1]{\left\{ #1\right\}}
\newcommand{\ssets}[1]{\{ #1\}}
\newcommand{\fwh}[1]{\; \left| \; #1 \right.}
\newcommand{\card}[1]{\left| #1 \right|}
\newcommand{\union}{\cup}
\newcommand{\map}{\longrightarrow}
\newcommand{\inters}{\cap}   
\newcommand{\then}{\Longrightarrow}
\newcommand{\vecc}[1]{\ensuremath{\bm{#1}}}
\newcommand{\opt}{\ensuremath{\mathrm{OPT}}}
\title{A Unifying Approximate Potential for\\ Weighted Congestion
Games\thanks{Part of this work was done while the authors were members of the
Operations Research group at Technical University of Munich, School of
Management, supported by the Alexander von Humboldt Foundation with funds from
the German Federal Ministry of Education and Research (BMBF). D.\ Poças was also
supported by FCT via LASIGE Research Unit, ref.\ UIDB/00408/2020 and ref.\ UIDP/00408/2020. \newline\indent
A preliminary version of this paper appeared in SAGT'20 \cite{gp2020_sagt}.}}
\author{Yiannis Giannakopoulos\thanks{Friedrich-Alexander-Universität Erlangen-Nürnberg, Department of Data Science  ({\tt
\href{mailto:yiannis.giannakopoulos@fau.de}{\nolinkurl{yiannis.giannakopoulos@fau.de}})
}}
	\and
		Diogo Poças\thanks{LASIGE, Faculdade de Ciências, Universidade de Lisboa
		({\tt \href{mailto:dmpocas@fc.ul.pt}{\nolinkurl{dmpocas@fc.ul.pt}})}}
}
\date{March 15, 2022}
\begin{document}
\maketitle
\begin{abstract}
We provide a unifying, black-box tool for establishing existence of approximate
equilibria in weighted congestion games and, at the same time, bounding their Price of
Stability. Our framework can handle resources with general costs---including, in particular, decreasing ones---and is formulated in terms of a
set of parameters which are determined via elementary analytic properties of the
cost functions.

We demonstrate the power of our tool by applying it to recover the recent result of
Caragiannis and Fanelli [ICALP'19] for polynomial congestion games; improve upon the
bounds for fair cost sharing games by Chen and Roughgarden [Theory Comput.\ Syst.,
2009]; and derive new bounds for nondecreasing concave costs. An interesting feature
of our framework is that it can be readily applied to \emph{mixtures} of different
families of cost functions; for example, we provide bounds for games whose resources
are conical combinations of polynomial and concave costs.

In the core of our analysis lies the use of a unifying approximate potential
function which is simple and general enough to be applicable to arbitrary congestion
games, but at the same time powerful enough to produce state-of-the-art bounds
across a range of different cost functions.
\end{abstract}

\section{Introduction}
\label{sec:intro}

Atomic congestion games are one of the most well-studied topics in \emph{algorithmic
game theory}~\cite{Roughgarden2016,2007a}. In their most general form, players have
weights and compete over a common set of resources; the cost of each resource is a
function of the total weight of the players that end up using it. As a result, they
can model a wide range of interesting applications including, e.g., network
routing~\cite{Roughgarden2007} and load balancing~\cite{Vocking2007a}, but also even
cost-sharing games (via the use of decreasing cost functions) like fair network
design~\cite{Tardos:2007aa}.

An important special case is that of \emph{unweighted} congestion games, where the
costs depend only on the \emph{number} of players that use each edge. In a seminal
paper, Rosenthal~\cite{Rosenthal1973a} proved that unweighted congestion games
always have (pure Nash) equilibria. A key tool in his derivation was the novel use
of a \emph{potential function}, which is able to capture the different players'
deviations in a very elegant and concise way. Then, the desired equilibrium is derived
as the minimizer of that function (over all feasible outcomes of the game).
This technique can also be viewed as an \emph{equilibrium refinement}, and
has been a very influential idea in game theory~\cite{Monderer1996a}. It allows us
not only to establish the existence of equilibria, but in many cases, this
special potential-minimizer equilibrium has additional desired properties.

Of particular importance to us in this paper, is that it has been the de facto
method for proving \emph{Price of Stability (PoS)} bounds in congestion games (see,
e.g.,~\cite[Ch.~18,~19]{2007a}). The PoS notion~\cite{Anshelevich2008a,Correa2004}
captures the minimum approximation ratio of the social cost, among all equilibria,
to the socially optimum outcome of the game (that might not be an equilibrium). In
other words, the PoS is the best-case counterpart of the notorious Price of Anarchy
(PoA) notion introduced by Koutsoupias and
Papadimitriou~\cite{Koutsoupias1999b,Papadimitriou2001a}

Unfortunately, though, it is a well-known fact that general \emph{weighted}
congestion games do \emph{not} always have equilibria and thus, do not admit a
potential function.
To alleviate this, a line of work has focused on designing \emph{approximate}
potential functions (see,
e.g.,~\cite{Chen2008,Hansknecht2014,Christodoulou2011a,cggs2018,Caragiannis:2019aa}):
the minimizer of such functions is guaranteed to be an approximate equilibrium (as
opposed to an \emph{exact} one that is given by Rosenthal's potential in the
unweighted case), while at the same time it can achieve a good approximation ratio
to the optimal social cost (providing, thus, an upper bound for the
approximate-equilibrium extension of the PoS notion). However, most of those prior
works use different approximate potentials, designed specially for the particular
cost-function model that each one studies.

Our goal in this paper is to provide a simple, high-level framework whose interface is
agnostic to the underlying potential function technicalities and which can readily
be instantiated for all resource costs at hand to derive meaningful bounds.

\subsection{Related Work}
\label{sec:related-work}

Following the seminal work of~\cite{Rosenthal1973a}, a long line of
results has been devoted to the (non)ex\-is\-tence of equilibria in weighted
congestion games. \cite{Goemans2005,Libman2001,Fotakis2005a} demonstrated that
equilibria might not exist even in very simple classes of games, including network
congestion games with quadratic cost functions and games where player weights
are either $1$ or $2$.
On the other hand, \cite{Fotakis2005a,Panagopoulou2007,Harks2012a} showed that
equilibria do exist in games with affine or exponential cost functions;
\cite{Fotakis2009a,Harks2012} proved the same for singleton games (where players can
only occupy single resources). Dunkel and Schulz~\cite{Dunkel2008} were able to
extend the nonexistence instance of Fotakis et al.~\cite{Fotakis2005a} to a hardness
gadget, in order to show that, deciding whether a congestion game with step cost
functions has an equilibrium, is a (strongly) NP-complete problem.

Regarding the existence of \emph{approximate} equilibria in general weighted
congestion games, \cite{cggpw2020} showed that games with $n$ players always have
$n$-approximate equilibria, and this guarantee is tight (up to logarithmic factors);
they also proved that the corresponding decision problem, i.e., of the existence of
$\tilde\varTheta(n)$-approximate equilibria, is NP-complete.

A lot of work has been focused on the important special case of \emph{polynomial}
congestions games, parameterized by the maximum degree $d$ of the cost functions.
Although, due to~\cite{Fotakis2005a} we already know that exact equilibria do not in
general exist in such games, Caragiannis et al.~\cite{Caragiannis2011} were the
first to show that $\alpha$-approximate equilibria do exist for $\alpha=d!$; this
factor was later improved to $\alpha=d+1$~\cite{Hansknecht2014,cggs2018} and
$\alpha=d$~\cite{Caragiannis:2019aa}.
As a matter of fact, Caragiannis and Fanelli~\cite{Caragiannis:2019aa} provide an
even more comprehensive result that, for any choice of a parameter $\delta\in[0,1]$,
simultaneously establishes the existence of $(d+\delta)$-approximate equilibria and
gives an upper bound of $\frac{d+1}{\delta+1}$ on their PoS. They achieve this by
designing an appropriate approximate potential function, tailored to polynomial
costs.
On the nonexistence front, \cite{Hansknecht2014} first gave
instances of very simple, two-player polynomial congestion games that do not have
$\alpha$-approximate equilibria, for $\alpha\approx 1.153$. This was recently
improved to $\alpha=\varOmega\left(\frac{\sqrt{d}}{\log d}\right)$ by Christodoulou et al.~\cite
{cggpw2020}, who also established NP-hardness of the corresponding existence
decision problem.

The work of Hansknecht et al.~\cite{Hansknecht2014} is very relevant for our
approach in this paper, since they also propose a ``generic'' approximate potential
function that can, in principle, be applied to general cost functions. They
instantiate it for polynomial costs to derive their aforementioned
existence of $(d+1)$-approximate equilibria. Additionally, they also state a result
about the existence of $\frac{3}{2}$-approximate equilibria in games with
nondecreasing concave costs; however, this proof in their paper is not complete.
Furthermore, \cite{Hansknecht2014} focuses just on the existence of approximate
equilibria, and thus it does not provide any PoS bounds.

Another well-studied class of congestion games is that of fair cost sharing, where
each resource has a constant initial cost which is split equally among the players
that use it. Thus, such games have \emph{decreasing} cost functions.
Finding the PoS for the special, undirected network version of such games is a
notorious open problem in the field (see, e.g.,
\cite{Anshelevich2008a,Fiat2006,Bilo2013,Bilo2014,Freeman2016}).
Very relevant for us is the work of Chen and Roughgarden~\cite{Chen2008} who showed
that general weighted fair cost sharing games always have $\alpha$-approximate equilibria
whose PoS is at most $O\left(\frac{\log W}{\alpha}\right)$, for any choice of parameter
$\alpha=\varOmega(\log w_{\max})$, where $w_{\max}$ is the maximum weight of any player
and $W$ the maximum possible load in any resource. They achieve this by designing a
special approximate potential function, tailored to the specific form of the cost functions.

\subsection{Our Results and Techniques}
\label{sec:results}

We propose a new approximate potential function
(see~\eqref{eq:potential-general-form}) for weighted congestion games with general
cost functions. In particular, our potential can be instantiated beyond the standard
model of polynomial cost functions and the common assumption of non-decreasing
monotonicity.
However, this potential is only used in the analysis part of our paper: we hide away
its specific form by hard-coding it within the proof of a unifying tool
(\cref{th:good_existence_easy}). Then, this tool can be used in a black-box way to
readily derive both existence of approximate equilibria, and bounds on their PoS.
This proof makes also use of a general, high-level lemma that can capture the
essence of the \emph{potential method} as a technique for deriving existence and PoS
bounds for approximate equilibria (\cref{lemma:potential-method}); we believe this
might be of independent interest, since in future work it could be used for
alternative potential functions, beyond our choice
of~\eqref{eq:potential-general-form} in this paper.

Our framework effectively works in two steps. Given a congestion game, first
one has to determine how \emph{good} its cost functions are with respect to two
simple, analytic properties (\cref{def:good-costs}). Then, the resulting
``goodness'' parameters can be plugged straight into our master theorem
(\cref{th:good_existence_easy}) to deduce the existence of an
$(\alpha,\beta)$-equilibrium; that is, an $\alpha$-approximate (pure Nash)
equilibrium whose social cost is at most a factor of $\beta$ away from the optimum.

We demonstrate the power of our tool by applying it to recover and improve prior
bounds on the existence of $(\alpha,\beta)$-equilibria for well-studied classes of
congestion games, as well as to derive novel results. The simplicity and the
algebraic nature of our tool allows us to produce fine-grained bounds in the form of
a parametric trade-off curve that describes the relation between the $\alpha$ and
$\beta$ parameters of the $(\alpha,\beta)$-equilibrium; in other words, all our results give a
\emph{continuum} of existence bounds.
Our bounds are summarized in~\cref{table:results}.
\begin{table}[t]
\centering
\footnotesize
\begin{tabular}{llll}
\toprule
& & \multicolumn{2}{c}{Our Results} \\
\cmidrule(lr{.75em}){3-4}
Cost functions & Previous Work & General & Extreme Points\\
\midrule
Polynomials of degree $\leq d$ & $\left(\lambda,\frac{d+1}{\lambda}\right)$, \hfill for $\lambda\in[d,d+1]$ & $\left(\lambda,\frac{d+1}{\lambda}\right)$, \hfill for $\lambda\in[d,d+1]$ & $(d,1+\frac{1}{d})$, $(d+1,1)$\\
& \ \hfill \cite{Caragiannis:2019aa}  & \ \hfill  [\cref{th:games-poly}] &  \\
Concave & $\left(\frac{3}{2},\infty \right)$ \hfill \cite{Hansknecht2014} & $\left(\lambda,\frac{\lambda}{\lambda-1}\right)$, \hfill for $\lambda\in\left[\frac{3}{2},2\right]$ & $\left(\frac{3}{2},3\right)$, $(2,2)$\\
&  & \ \hfill  [\cref{th:games-concave}] & \\
Polynomials + Concave & N/A & $\left(\lambda,1+\frac{d+1}{\lambda}\right)$, \hfill for $\lambda\in[d,d+1]$ & $(d,2+\frac{1}{d})$, $(d+1,2)$ \\
&  &  \ \hfill [\cref{th:games-concave-poly}] & \\
Fair cost sharing & $\left(\lambda,1+\frac{2\log_2(1+W)}{\lambda}\right),$ & $\left(\varTheta(\ln w_{\max})+\lambda,1+\frac{\ln W}{\lambda}\right)$, & $\left(\varTheta(\ln w_{\max}),1+\ln W\right)$,\\
& for $\lambda= \varOmega(\ln w_{\max})$ \hfill \cite{Chen2008} & for $\lambda\geq 1$ \hfill [\cref{thm:fair-cost-pne}] & $(\varTheta(\ln W),\varTheta(1))$\\
\bottomrule
\end{tabular}
\caption{Our main results on the existence of $(\alpha,\beta)$-equilibria for different
cost models. For polynomials of degree $d$ we recover the result of
\cite{Caragiannis:2019aa}. For fair costs our results improve those of
\cite{Chen2008} and for concave costs we extend those of~\cite{Hansknecht2014}. For
mixtures of different cost functions, namely polynomial and concave, our results are
novel.}
\label{table:results}
\end{table}

More specifically, first (\cref{th:games-poly}) we rederive the recent bounds
of~\cite{Caragiannis:2019aa} for polynomial congestion games, in a more ``clean'',
high-level way.
Then (\cref{thm:fair-cost-pne}), we improve the $\alpha,\beta$ parameters on the
$(\alpha,\beta)$-equilibrium existence results of~\cite{Chen2008} for fair
cost-sharing games (a more detailed comparison can be seen in~\cref{fig:fair-cost}).
Furthermore, we derive new results for (nondecreasing) concave costs: we show that
$(\lambda,\frac{\lambda}{\lambda-1})$-equilibria always exist, for all
$\lambda\in[\frac{3}{2},2]$ (\cref{th:games-concave}). The special corner case of a
$(\frac{3}{2},3)$-equilibrium is compatible, thus, with the
$\frac{3}{2}$-approximate equilibrium existence stated in~\cite{Hansknecht2014}.

Another interesting characteristic of our tool is its \emph{modularity}: it
can readily combine different cost functions to give bounds for more complex
congestion games (see~\cref{def:good-games}). For example, we prove that games with
cost functions that are conical combinations of $d$-degree polynomials and concave
costs, always have $\left(\lambda,1+\frac{d+1}{\lambda}\right)$-equilibria, where $\lambda$
ranges in $[d,d+1]$ (\cref{th:games-concave-poly}).

Finally, an added advantage of our black-box method is that it also
results in arguably simpler and more streamlined proofs for the existence and PoS bounds.

Before concluding the overview of our results, we want to elaborate a bit more on
the comparison to the potential approach of Hansknecht et al.~\cite{Hansknecht2014}.
Although~\cite{Hansknecht2014} does not deal with PoS bounds, as far as existence of
approximate equilibria is concerned, their paper is rather similar in principle to
ours. They propose a general potential function which is based on a discrete
interpretation of the cost function's integral, which corresponds to the first
component of our potential in~\eqref{eq:potential-general-form}. We take a different
approach by using directly the \emph{actual} integral, and also adding an extra term
that corresponds to a weighted average of the costs of the players' weights. In that
way, we avoid a lot of the intricate technicalities that are involved with the
discrete arguments (e.g., orderings of the weights) in~\cite{Hansknecht2014}, making
the application of our potential (via our high-level tool
of~\cref{th:good_existence_easy}) more ``tractable'' for a wider range of cost
functions.

\section{Model and Notation}
\label{sec:model}
We use $\R_+$ to denote the set of nonnegative real numbers.

In a \emph{(weighted) congestion game} $\mathcal G$
there are finite, nonempty sets of \emph{players} $N$ and \emph{resources} $E$.
Let $n=\card{N}$.
Each player $i\in N$ has a \emph{weight} $w_i\in\R_{+}$ and a \emph{strategy set}
$S_i\subseteq 2^E$. We use $w_{\min}=\min_{i\in N}  w_i$ and $w_{\max}=\max_{i\in N} w_i$ for the minimum and
maximum player weights, respectively, and for a subset of players $I\subseteq N$, we use $w_I=\sum_{i\in
I} w_i$ to denote the sum of their weights. For the special case of $w_{\min}=w_{\max}=1$, that is, if all weights are $1$, we say that $\mathcal G$ is \emph{unweighted}.

Associated with each resource $e\in E$ is a \emph{cost function} $c_e: \R_{+}\map
\R_{+}$. In general, we will make no extra assumptions on the cost functions.
However, important special cases, that we will also study as applications of the
main tool of our paper, include \emph{polynomial congestion games} of degree $d$,
for $d\geq 1$ integer, and \emph{fair cost sharing games}. In the former, the cost
functions are polynomials with nonnegative coefficients and degree at most $d$; in
the latter, cost functions are (decreasing) of the form $c_e(x)=\frac{a_e}{x}$ where
$a_e$ is a positive real.

A (pure) \emph{strategy profile} (or \emph{outcome}) is a choice of strategies
$\vecc{s}= (s_1, s_2,..., s_n)\in \vecc{S}={S}_1\times \cdots \times {S}_n$. We use
the standard game-theoretic notation $\vecc{s}_{-i}=(s_1,\ldots, s_{i-1},\allowbreak
s_{i+1},\allowbreak \ldots s_n)$, $\vecc{S}_{-i}={S}_1\times \cdots \times S_{i-1}
\times S_{i+1} \times \cdots \times S_n$. In that way, for example, we can denote
$\vecc{s}=(s_i,\vecc{s}_{-i})$. Given a profile $\vecc{s}\in\vecc{S}$, we define the
\emph{load} $x_e(\vecc{s})$ of resource $e$ as the total weight of players that use
resource $e$ at outcome $\vecc s$, i.e., $x_e(\vecc{s})=w_{N_e(\vecc s)}= \sum_{i\in
N: e\in s_i} w_i$, where $N_e(\vecc s)$ is the set of players using $e$.
We will use $W=\sum_{i\in N} w_i$ to denote the maximum possible load of any
resource.
The \emph{cost} of player $i$ is
defined by $C_i(\vecc{s})=\sum_{e\in s_i} c_e(x_e(\vecc{s}))$.
The \emph{social cost} of a strategy profile $\vecc{s}$ is the weighted sum of the players' costs
$$C(\vecc{s})=\sum_{i\in N}w_i \cdot C_i(\vecc{s}) = \sum_{e\in E}
x_e(\vecc{s}) \cdot c_e(x_e(\vecc{s})).$$
We use $\opt(\mathcal G)=\min_{\vecc{s}\in S} C(\vecc{s})$ to denote the \emph{optimum social
cost} over all outcomes.

An outcome $\vecc s$ is an \emph{$\alpha$-approximate} (pure Nash)
\emph{equilibrium}, for $\alpha\geq 1$, if 
\begin{equation}
\label{eq:pne-def}
C_i(\vecc{s})\leq \alpha \cdot C_i(s'_i,\vecc{s}_{-i})
\qquad\text{for all}\;\; i\in N,\; s'_i\in S_i
\end{equation}
That is, no player can unilaterally deviate from $\vecc s$ and improve her cost by
more than a factor of $\alpha$. Notice that for the special case of $\alpha=1$ we get the definition of the standard, \emph{exact} pure Nash equilibrium.
We denote the set of all $\alpha$-equilibria of $\mathcal G$ by $\mathrm{NE}_\alpha(\mathcal G)$
Then, the \emph{$\alpha$-approximate Price of Stability ($\alpha$-PoS)} of $\mathcal G$ is the social cost of the best-case
Nash equilibrium over the optimum social cost:
\begin{equation}
\label{eq:pos-def}
 \mathrm{PoS}_\alpha(\mathcal G) =
\min_{\vecc{s}\in\mathrm{NE}_\alpha(\mathcal G)}\frac{C(\vecc{s})}{\opt(\mathcal G)}.
\end{equation}
For $\alpha=1$ we get the standard definition of the Price of Stability (PoS) for exact equilibria~\cite{Anshelevich2008a}.
We combine the notions of an approximate equilibrium with approximating the optimum social cost in the following definition:
\begin{definition}[$(\alpha,\beta)$-equilibrium]
Fix a congestion game $\mathcal G$. A strategy profile $\vecc{s}$ is an $(\alpha,\beta)$-equilibrium if it is an $\alpha$-approximate equilibrium of $\mathcal G$ (see~\eqref{eq:pne-def}) and its social cost is at most $\beta$ times the optimal cost of $\mathcal G$, i.e., $C(\vecc s) \leq \beta \cdot \opt(\mathcal G)$.
\end{definition}
Notice that if a game has an $(\alpha,\beta)$-equilibrium then, due to~\eqref{eq:pos-def}, its $\alpha$-PoS is at most $\beta$. 

\subsection{Equivalent Cost Functions} 
\label{sec:equivalent-costs}
It is not difficult to see that, in any
weighted congestion game, the cost functions of each resource are actually evaluated
on finitely many points: although our model assumes $c_e$ to be
defined over the entire $\R_{+}$, its values outside the domain $\ssets{x_e(\vecc
s)\fwh{\vecc s\in \vecc S}}$ are irrelevant. In particular, this domain is included
within the set of different sums of weights
$$\mathcal{W}=\sset{\left.\sum_{i\in N}y_i\cdot w_i\;\right| \; y_i\in\ssets{0,1},\;
i\in N}.$$ This means that one only needs to define costs on
at most $\card{\mathcal{W}}\leq 2^n$ different values: any two games whose costs
coincide on $\mathcal W$ are equivalent.

However, it is still convenient to treat our costs as functions over $\R_{+}$. First, because this allows for simple and succinct representations. But of
particular importance to us, is also the fact that our main tool
(\cref{th:good_existence_easy}) can be applied to all \emph{integrable} cost
functions (so that~\cref{def:good-costs} can be utilized). From the above discussion,
it should be obvious that any congestion game has (infinitely) many equivalent
representations, that is, different extensions from $\mathcal W$
to $\R_{+}$. Such an extension can always be done in a way that $c_e$ is an integrable
function (since $\mathcal W$ is finite). 

It is interesting to point out here that different representations can potentially
give different existence and PoS bounds via our tool. Although we do not deal with
this feature for most of the paper, it is important for our fair cost sharing results
(\cref{sec:fair-costs-simple}); since function $x\mapsto 1/x$ is not integrable over
the interval $[0,w_{\min})$ (and as a matter of fact, not even defined on $x=0$) we
have the freedom, according to the discussion above, to redefine it in any way we
want on $[0,w_{\min})$, so that it is a well-defined, integrable function over
$\R_{+}$.

\section{The Main Tool}
\label{sec:basic-tool}
In this section we present our framework for establishing existence of
$(\alpha,\beta)$-equilibria in weighted congestion games with general cost
functions.
We begin with the following lemma, that tries to distil and abstract the potential
method technique in congestion games. Specialized or restricted forms of it have
essentially been used, even if not explicitly stated, in multiple works in the past
(see, e.g.,~\cite{Chen2008,Hansknecht2014,Caragiannis:2019aa}). It can be seen as a
more fine-grained version of~\cite[Lemma~4.1]{cggs2018}, although some extra care is
needed to adapt it to the more abstract setting of our paper and utilize its full
power. 

\begin{lemma}[Potential Method]
\label{lemma:potential-method}
Fix a congestion game. Assume that, for each resource $e$, there exist positive reals $\alpha_{1,e}, \alpha_{2,e}, \beta_{1,e}, \beta_{2,e}$, and a
function $\phi_e: 2^N\map \R$ such that $\phi_e(\emptyset)=0$ and
\begin{equation}
\label{eq:potential-method-cond-1}
\alpha_{1,e} \leq \frac{\phi_e(I\union\ssets{i})-\phi_e(I)}{w_i\cdot c_e(w_I+w_i)} \leq \alpha_{2,e}
\qquad\text{for all}\;\; i\in N,\; I\subseteq N\setminus \ssets{i};
\end{equation}
\begin{equation}
\label{eq:potential-method-cond-2}
\beta_{1,e}\leq\frac{\phi_e(I)}{w_I\cdot c_e(w_I)} \leq \beta_{2,e}
\qquad\text{for all}\;\; \emptyset\neq I \subseteq N.
\end{equation} 
Then the game has an $(\alpha,\beta)$-equilibrium with 
$$
\alpha= \max_{e\in E}\frac{\alpha_{2,e}}{\alpha_{1,e}}
\quad\text{and}\quad
\beta = 
\frac{\max_{e\in E} \beta_{2,e}/\alpha_{1,e}}{\min_{e\in E} \beta_{1,e}/\alpha_{1,e}}.
$$
\end{lemma}
\begin{proof}
Define function
$
\varPhi(\vecc s)=\sum_{e\in E}\frac{1}{\alpha_{1,e}}\phi_e(N_e(\vecc s))
$
over all feasible outcomes. We will show that $\varPhi$ can serve as a desired \emph{approximate potential} function for our game; that is, for any profiles $\vecc{s},\vecc{s}'$ and any player $i$, it satisfies:
\begin{align}
\varPhi(\vecc s) \leq \varPhi(s_i',\vecc s_{-i})
		&\quad \then \quad C_i(\vecc s) \leq \alpha\cdot C_i(s_i',\vecc s_{-i}) \label{eq:approx_potential_existence}\\
\varPhi(\vecc s) \leq \varPhi(\vecc s') 
		& \quad \then \quad C(\vecc s) \leq \beta \cdot C(\vecc s'). \label{eq:approx_potential_PoS}
\end{align}
This would be enough to establish our lemma: any (global) minimizer of $\varPhi$ is
an $\alpha$-approximate equilibrium, due to~\eqref{eq:approx_potential_existence}, and at the same time,
due to~\eqref{eq:approx_potential_PoS}, its social cost is within a factor of
$\beta$ from the social cost of any other profile (and, thus, from the optimal one).
Notice also, that such a minimizer always exists, since the set $\vecc{S}$ of feasible outcomes is finite.

For~\eqref{eq:approx_potential_existence} first, denote for simplicity $N_e=N_e(\vecc s)$, $x_e=x_e(\vecc{s})$ and
$N_e'=N_e(s_i',\vecc s_{-i})$, $x_e'=x_e(s_i',\vecc s_{-i})$ for all facilities $e$. Then, we have
\begin{align*}
\varPhi(s_i',\vecc s_{-i})-\varPhi(\vecc s)
	&= \sum_{e\in E} \frac{1}{\alpha_{1,e}}\left[ \phi_e(N_e')-\phi_e(N_e)\right]\\
	&= \sum_{e\in s_i'\setminus s_i}\frac{1}{\alpha_{1,e}}\left[ \phi_e(N_e\union\ssets{i})-\phi_e(N_e)\right] + \sum_{e\in s_i\setminus s_i'}\frac{1}{\alpha_{1,e}}\left[ \phi_e(N_e\setminus\ssets{i})-\phi_e(N_e)\right]\\
	&\leq \sum_{e\in s_i'\setminus s_i}\frac{\alpha_{2,e}}{\alpha_{1,e}} w_i c_e(x_e+w_i) - \sum_{e\in s_i\setminus s_i'} w_i c_e(x_e)\\
	&\leq w_i\left[\alpha\sum_{e\in s_i'\setminus s_i} c_e(x_e+w_i) - \sum_{e\in s_i\setminus s_i'}  c_e(x_e)\right]\\
	&\leq w_i\left[\alpha \left(\sum_{e\in s_i'\setminus s_i} c_e(x_e+w_i) +\sum_{e\in s_i'\inters s_i} c_e(x_e)\right) - \left(\sum_{e\in s_i\setminus s_i'} c_e(x_e) +\sum_{e\in s_i'\inters s_i} c_e(x_e)\right)\right]\\
	&= w_i\left[\alpha C_i(s_i',\vecc s_{-i}) - C_i(\vecc s) \right].
\end{align*}
The first inequality holds due to~\eqref{eq:potential-method-cond-1}; the second due
to the definition of $\alpha$; and the third one because $\alpha\geq 1$. The fact
that the cost functions are nonnegative is a critical component in all of them as
well. The chain of inequalities above demonstrate that, if $\varPhi(s_i',\vecc s_{-i})-\varPhi(\vecc s)$ is nonnegative then $\alpha C_i(s_i',\vecc s_{-i}) - C_i(\vecc s)$ is nonnegative, thus proving~\eqref{eq:approx_potential_existence}.

For~\eqref{eq:approx_potential_PoS} next, denote $N_e=N_e(\vecc s)$, $x_e=x_e(\vecc{s})$ and
$N_e'=N_e(\vecc s')$, $x_e'=x_e(\vecc s')$. Then, we have:
\begin{align*}
\varPhi(\vecc s')-\varPhi(\vecc s) 
	& = \sum_{e\in E}\frac{1}{\alpha_{1,e}}\phi_e(N_e')-\sum_{e\in
E}\frac{1}{\alpha_{1,e}}\phi_e(N_e)\\ 
	&\leq \sum_{e\in E}\frac{\beta_{2,e}}{\alpha_{1,e}}x_e'c_e(x_e')
-\sum_{e\in E}\frac{\beta_{1,e}}{\alpha_{1,e}}x_ec_e(x_e)\\ 
	&\leq \max_{e\in E}\frac{\beta_{2,e}}{\alpha_{1,e}}\cdot \sum_{e\in E}x_e'c_e(x_e')
-\min_{e\in E}\frac{\beta_{1,e}}{\alpha_{1,e}}\cdot \sum_{e\in E}x_ec_e(x_e)\\ 
	&= \max_{e\in E}\frac{\beta_{2,e}}{\alpha_{1,e}}\cdot C(\vecc s')
-\min_{e\in E}\frac{\beta_{1,e}}{\alpha_{1,e}}\cdot C(\vecc s)\\ 
	&=\min_{e\in E}(\beta_{1,e}/\alpha_{1,e}) \cdot\left[\beta C(\vecc s') -  C(\vecc s) \right],
\end{align*}
where for the first inequality we deployed~\eqref{eq:potential-method-cond-2}. The
chain of inequalities above demonstrate that if $\varPhi(\vecc s')-\varPhi(\vecc s)$
is nonnegative then $\beta C(\vecc s') -  C(\vecc s)$ is nonnegative as well, this
establishing~\eqref{eq:approx_potential_PoS}.
\end{proof}

We continue with defining a critical notion that will act as the medium to utilize
our main black-box tool in~\cref{th:good_existence_easy}. It involves a set of
parameters, that determine how ``well'' a given cost function behaves with respect
to two specific, simple analytic properties (namely~\eqref{eq:good_1}
and~\eqref{eq:good_2}). These properties can be interpreted as bounds on the average
of the cost function over continuous intervals.

\begin{definition}[Good Cost Functions]
\label{def:good-costs}
Fix a  congestion game $\mathcal G$. A function $c:\R_+\map\R_+$ will be called \emph{$(\alpha_1,\alpha_2,\beta_1,\beta_2)$-good} (with respect to $\mathcal G$), for $\alpha_1,\alpha_2,\beta_1,\beta_2>0$, if there exists a nonnegative constant $\xi$ such that,
for all $x\in\ssets{0}\union[w_{\min},W]$, $w\in[w_{\min},w_{\max}]$:
\begin{equation}
\label{eq:good_1}
 \alpha_1\cdot c(x+w) - \xi \cdot c(w)
\;\leq\;
\frac{1}{w}\int_x^{x+w} c(t)\,dt 
\;\leq\; 
\alpha_2 \cdot c(x+w) - \xi \cdot c(w)
\end{equation}
and for all $x\in [w_{\min},W]$:
\begin{equation}
\label{eq:good_2}
 \beta_1\cdot c(x) - \xi \cdot c_{\min}(x)
\;\leq\; 
\frac{1}{x}\int_0^{x} c(t)\,dt
\;\leq\; 
\beta_2\cdot c(x) - \xi \cdot c_{\max}(x),
\end{equation}
where $c_{\min}(x)=\min_{y\in[w_{\min},x]}c(y)$, $c_{\max}(x)=\max_{y\in[w_{\min},x]}c(y)$.
\end{definition}

\begin{definition}[Good Games]
\label{def:good-games}
A congestion game will be called $\ssets{(\alpha_{1,j},\alpha_{2,j},\beta_{1,j},\beta_{2,j})}_{j\in J}$-good if any cost function is a conical combination of such good functions. Formally, for any $e\in E$ there exists a nonempty $J_e\subseteq J$ and nonnegative constants $\ssets{\lambda_{e,j}}_{j\in J_e}$, such that 
\begin{equation*}
c_e(t)= \sum_{j\in J_e} \lambda_{e,j} c_j(t)
\end{equation*}
where, for all $j\in J$, $c_j$ is a  $(\alpha_{1,j},\alpha_{2,j},\beta_{1,j},\beta_{2,j})$-good function (see~\cref{def:good-costs}).
\end{definition}

\begin{remark}
\label{remark:conical-to-simple}
Notice that an important special case of~\cref{def:good-games} is when $J=E$, $J_e=\{e\}$, and
$\lambda_{e,e}=1$, meaning that the actual cost functions
of the game are good themselves. As a matter of fact, it is not hard to see that any
good game $\mathcal G$ can be transformed to a strategically equivalent one
$\mathcal G'$ that has that property.  First, replace each resource $e$ of $\mathcal
G$ with a gadget of ``parallel'' resources $\ssets{(e,j)\;|\;j\in J_e}$, each having a
cost function of $c_{(e,j)}(t)=\lambda_{e,j}c_j(t)$; this results in a strategically
equivalent game $\mathcal G'$ with resources $E'=\ssets{(e,j)\;|\;e\in E,\; j\in
J_e}$. Next, just observe that~\cref{def:good-costs} is invariant under nonnegative
scalar multiplication: since functions $c_j$ satisfy conditions \eqref{eq:good_1}
and \eqref{eq:good_2}, so do functions $\lambda_{e,j}\cdot c_j$ that are exactly the
cost functions of the new game $\mathcal G'$.
\end{remark}
\begin{remark}[Increasing Good Functions]
\label{remark:good-increasing-simpler}
If a cost function is nondecreasing, then~\eqref{eq:good_2} can be replaced by the (stronger, sufficient) condition:
\begin{equation}
\label{eq:good_2_increasing}
\beta_1 c(x) \;\leq\; \frac{1}{x}\int_{0}^x c(t)\,dt \;\leq\; (\beta_2-\xi) c(x),
\tag{\ref*{eq:good_2}$^\prime$} 
\end{equation}
since $0 \leq c(y) \leq c(x)$ for any $y\in[w_{\min},x]$.
\end{remark}

Now we are ready to state our main tool. This is essentially the interface of our
entire framework: under the hood it uses a specific potential function form
(see~\eqref{eq:potential-general-form}), but its statement involves only the
goodness parameters of the cost functions, as defined above. In that way, one can
readily derive meaningful bounds about the existence of $(\alpha,\beta)$-equilibria
in a black-box way, just by studying the simple analytic properties given
in~\eqref{def:good-costs} and the plugging the parameters in the theorem below:

\begin{theorem}
\label{th:good_existence_easy}
Any $\sset{(\alpha_{1,j},\alpha_{2,j},\beta_{1,j},\beta_{2,j})}_{j\in J}$-good congestion game has an $(\alpha,\beta)$-equilibrium with
$$
\alpha = \max_{j\in J}\frac{\alpha_{2,j}}{\alpha_{1,j}}
\qquad\text{and}\qquad
\beta = \frac{\max_{j\in J} \beta_{2,j}/\alpha_{1,j}}{\min_{j\in J} \beta_{1,j}/\alpha_{1,j}}.
$$
\end{theorem}
\begin{proof}
First notice that, by~\cref{remark:conical-to-simple}, it is without loss to assume that $J=E$ and that any cost function $c_e$, $e\in E$, is $(\alpha_{1,e},\alpha_{2,e},\beta_{1,e},\beta_{2,e})$-good. Denote by $\xi_e$ (a choice of) the parameter $\xi$ for which resource $e$ satisfies~\cref{def:good-costs}.

We will then show that functions
\begin{equation}
\label{eq:potential-general-form}
\phi_e(I) = \int_{0}^{w_I} c_e(t)\, dt + \xi_e \sum_{i\in I} w_i c_e(w_i)
\end{equation}
satisfy the conditions of~\cref{lemma:potential-method},

Fix some resource $e\in E$, a player $i$ and a subset $I\subseteq
N\setminus\ssets{i}$ of remaining players. For simplicity, from now on we drop the
$e$ subscripts and also denote $w=w_i$ and $x=w_I$. Then,
\begin{align*}
\phi(I\union\ssets{i}) - \phi(I)
	&= \int_{0}^{x+w} c_e(t)\, dt - \int_{0}^{x} c_e(t)\, dt
	+ \xi_e\left(\sum_{j\in I} w_j c_e(w_j) - \sum_{j\in I\union\ssets{i}} w_j c_e(w_j) \right)\\
	&=\int_{x}^{x+w} c(t)\, dt + \xi w c(w).
\end{align*}
So, by deploying~\eqref{eq:good_1}, it is not difficult to see that
$$
\alpha_1 c(x+w)\leq \frac{1}{w}\left[\phi(I\union\ssets{i}) - \phi(I)\right] \leq \alpha_2 c(x+w),
$$
and thus condition~\eqref{eq:potential-method-cond-1} of~\cref{lemma:potential-method} is indeed satisfied.

Next, observe that since $w_j\in[w_{\min},w_{\max}]$ for all $j\in I$, and $x=\sum_{j\in I} w_j$, we have the bounds
\begin{equation}
\label{eq:good-proof-weighted-average}
c_{\min}(x)
\leq \min_{j\in I} c(w_j)
\leq \frac{1}{x}\sum_{j\in I}w_jc(w_j) 
\leq \max_{j\in I} c(w_j)
\leq c_{\max}(x),
\end{equation}
where the first and the last inequalities hold due to the fact that $\ssets{w_j\fwh{j\in I}}\subseteq [w_{\min},x]$.
Assuming $I\neq \emptyset$, we have that $x\in
[w_{\min},W]$ and so we can use~\eqref{eq:good-proof-weighted-average} and~\eqref{eq:good_2} to bound $\frac{1}{x}\phi(I)$ from below and above by:
$$
\beta_1 c(x) \leq
\frac{1}{x}\phi(I)= \frac{1}{x}\int_{0}^x c(t)\, dt+\xi \frac{1}{x}\sum_{j\in I}w_jc(w_j) 
\leq \beta_2 c(x).
$$
Thus, condition~\eqref{eq:potential-method-cond-2} of~\cref{lemma:potential-method} is also satisfied.

\end{proof}

\section{Applications}
\label{sec:applications}
In this section we present several applications of our black-box
\cref{th:good_existence_easy}, that demonstrate both its power and simplicity. In
accordance to the nature of that tool, they all share a common structure: first, we
prove lemmas describing the right goodness parameters (according
to~\cref{def:good-costs}) for each special cost function of interest
(see~\cref{lemma:polynomials-good,lemma:constant-good,lemma:concave-good,lem:fair-costs-good});
then, we plug them in~\cref{th:good_existence_easy} to derive our bounds
(see~\cref{th:games-poly,th:games-concave,thm:fair-cost-pne}).

\subsection{Polynomial Costs}
\label{sec:poly-simple}

We start with polynomial cost functions, arguably the most studied setting in
congestion games. We recover the result from Caragiannis and
Fanelli~\cite{Caragiannis:2019aa} that, for polynomials of degree at most $d$ with
nonnegative coefficients, there exist $(d+\delta)$-approximate equilibria with
social cost at most $\frac{d+1}{d+\delta}$ times the optimum, for any
$\delta\in[0,1]$. This is the currently best known guarantee of
$(\alpha,\beta)$-equilibria for polynomial cost functions. Let us begin by analysing
the goodness parameters of each monomial.

\begin{lemma}
\label{lemma:polynomials-good}
Any monomial of degree $d\geq 1$ is $\left(\mu,1,\frac{1}{d+1},\mu\right)$-good, for any $\mu\in[\frac{1}{d+1},\frac{1}{d}]$.
\end{lemma}
\begin{proof}
Fix a degree $d\geq 1$. We will show that the function $c(x)=x^d$ satisfies conditions~\eqref{eq:good_1} and~\eqref{eq:good_2} with 
$$
\alpha_1=\xi+\frac{1}{d+1},\qquad 
\alpha_2=1,\qquad 
\beta_1= \frac{1}{d+1},\qquad 
\beta_2=\xi+\frac{1}{d+1},
$$
for all $\xi\in[0,\frac{1}{d(d+1)}]$.
Then, performing the change of variables $\mu=\xi+\frac{1}{d+1}$ establishes our lemma, since $\mu\in[0+\frac{1}{d+1},\frac{1}{d(d+1)}+\frac{1}{d+1}]=[\frac{1}{d+1},\frac{1}{d}]$.

To prove the bounds in $\alpha_1$, $\alpha_2$, we are interested in the quantity
\[a(w,x)=\frac{1}{w}\int_x^{x+w}c(t)dt+\xi c(w)=\frac{1}{(d+1)w}\left((x+w)^{d+1}-x^{d+1}\right)+\xi w^d.\]
By applying the binomial expansion rules, and collecting similar terms, we can further write
\begin{align}
a(w,x) &=\frac{1}{(d+1)w}\left(w^{d+1}+\sum_{j=1}^d\binom{d+1}{j}x^jw^{d+1-j}+x^{d+1}-x^{d+1}\right)+\xi w^{d+1} \nonumber\\
	&=\frac{1}{d+1}\left(w^d+\sum_{j=1}^d\binom{d+1}{j}x^jw^{d-j}\right)+\xi w^d \nonumber\\
	&=\left(\frac{1}{d+1}+\xi\right)w^d+\sum_{j=1}^d\frac{1}{d+1}\binom{d+1}{j}x^jw^{d-j} \nonumber\\
	&=\left(\xi+\frac{1}{d+1}\right)w^d+\sum_{j=1}^d\frac{1}{d-j+1}\binom{d}{j}x^jw^{d-j},\label{eq:monomialexpansion1}
\end{align}
where in the last step we simply use the fact that, for $1\leq j\leq d$,
$\frac{1}{d+1}\binom{d+1}{j}=\frac{1}{d-j+1}\binom{d}{j}$. We would like to get
upper and lower bounds on $a(w,x)$ involving $c(x+w)$, which can be written as
\begin{equation}
\label{eq:monomialexpansion2}
c(x+w)=(x+w)^d=w^d+\sum_{j=1}^d\binom{d}{j}x^jw^{d-j}.
\end{equation} 
By comparing the coefficients of \eqref{eq:monomialexpansion1} and
\eqref{eq:monomialexpansion2}, we get that \eqref{eq:good_1} is satisfied with
\begin{align*}
\alpha_1 &= \min\left\{\xi+\frac{1}{d+1},\min_{j=1,\dots,d}\left\{\frac{1}{d-j+1}\right\}\right\} 
	= \min\left\{\xi+\frac{1}{d+1},\frac{1}{d}\right\} 
	=  \xi+\frac{1}{d+1}\\
\alpha_2 &= \max\left\{\xi+\frac{1}{d+1},\max_{j=1,\dots,d}\left\{\frac{1}{d-j+1}\right\}\right\} 
	= \max\left\{\xi+\frac{1}{d+1},1\right\} =1,
\end{align*}
where to compute the maxima and minima we used the fact that
$\xi+\frac{1}{d+1} \leq \frac{1}{d(d+1)}+\frac{1}{d+1} = \frac{1}{d} \leq 1$,
due to the assumptions that $\xi \leq \frac{1}{d(d+1)}$ and $d\geq 1$.

For the bounds in $\beta_1,\beta_2$, since $x^d$ is nondecreasing we can use the simpler condition \eqref{eq:good_2_increasing}. Then, we only have to observe that
\begin{align*}\frac{1}{x}\int_0^xc(t)dt
&=\frac{1}{d+1}x^d=\frac{1}{d+1}c(x)\\
\intertext{and}
\frac{1}{x}\int_0^xc(t)dt+\xi c(x)
&=\frac{1}{d+1}x^d+\xi x^d=\left(\xi+\frac{1}{d+1}\right)c(x).
\end{align*}
\end{proof}

For the special case of constant cost functions, i.e., $0$-degree monomials, it is not difficult to get the following:
\begin{lemma}
\label{lemma:constant-good}
Any constant function is $\left(1,1,1,1\right)$-good.
\end{lemma}
\begin{proof}
Follows directly from \cref{def:good-costs} by taking $\xi=0$: for any constant function $c(x)=c$ we have
$$c(x+w)=c(x)=\frac{1}{w}\int_x^{x+w}c(t)dt=\frac{1}{x}\int_0^xc(t)dt=c.$$
\end{proof}

\begin{theorem}
\label{th:games-poly}
Any weighted polynomial congestion game of degree $d\geq 1$ has an $(\lambda,\frac{d+1}{\lambda})$-equilibrium, 
for any $\lambda\in[d,d+1]$.
\end{theorem}

\begin{proof}
Fix a maximum degree $d\geq 1$ and a parameter $\lambda\in[d,d+1]$.
Utilizing \cref{lemma:polynomials-good} with $\mu =\frac{1}{k+1}$ and \cref{lemma:constant-good}, we can see that monomials of degree $k=0,\dots,d-1$ are $(\frac{1}{k+1},1,\frac{1}{k+1},\frac{1}{k+1})$-good; and utilizing~\cref{lemma:polynomials-good} with $\mu=\frac{1}{\lambda}$ we get that the monomial of degree $d$ is $(\frac{1}{\lambda},1,\frac{1}{d+1},\frac{1}{\lambda})$-good.

Since any polynomial of degree (at most) $d$ is a conical combination of monomials of degree $k=0,1,\dots,d$, in light of~\cref{def:good-games}, 
we can deduce that our game is $\left\{(\alpha_{1,k},\alpha_{2,k},\beta_{1,k},\beta_{2,k})\right\}_{k=0,\ldots,d}$-good, with
\begin{equation*}
\label{eq:poly-games-good-parameters}
\alpha_{1,k}=
\begin{cases}
\frac{1}{\lambda}, &k=d,\\ 
\frac{1}{k+1}, &k<d;
\end{cases}
\quad
\alpha_{2,k}=1;
\quad
\beta_{1,k}=
\begin{cases}
\frac{1}{d+1}, &k=d,\\ 
\frac{1}{k+1}, &k<d;
\end{cases}
\quad
\beta_{2,k}=
\begin{cases}
\frac{1}{\lambda}, &k=d,\\ 
\frac{1}{k+1}, &k<d.
\end{cases}
\end{equation*}
Thus, by~\cref{th:good_existence_easy} we conclude that our game has an $(\alpha,\beta)$-equilibrium with
\begin{align*}
\alpha 
 &=\max_{0\leq k\leq d}\frac{\alpha_{2,k}}{\alpha_{1,k}}
  =\max\left\{1,2,\dots,d,\lambda\right\}
  =\lambda,
\intertext{and}
\beta 
	&= \frac{\max\limits_{0\leq k\leq d}\frac{\beta_{2,k}}{\alpha_{1,k}}}{\min\limits_{0\leq k\leq d} \frac{\beta_{1,k}}{\alpha_{1,k}}}
	= \frac{1}{\min\left\{1,\dots,1,\frac{1/(d+1)}{1/\lambda}\right\}} 
	= \max\sset{1,\frac{d+1}{\lambda}} = \frac{d+1}{\lambda}.
\end{align*}
\end{proof}

The parameter $\lambda$ quantifies the trade-off curve between the approximation
guarantee on the existence of $\alpha$-approximate equilibria and their PoS. At one
extreme case $\lambda=d+1$, we get that $\alpha=d+1$ and $\beta=1$; in other words,
there always exist $(d+1)$-approximate equilibria with an optimal PoS of $1$ (as a
matter of fact, from~\cite{cggs2018} we already know that every social optimum is
itself a $(d+1)$-approximate equilibrium). At the other extreme case $\lambda=d$, we
get that
\[\alpha=d,\qquad\beta= \frac{d+1}{d}=1+\frac{1}{d};\]
in other words, there always exist $d$-approximate equilibria with PoS at most $1+\frac{1}{d}$.

\subsection{Concave Costs}
\label{sec:concave-simple}

We now look at nondecreasing concave cost functions. The best known result in this
setting is due to Hansknecht et al.~\cite{Hansknecht2014}, who state that
$3/2$-approximate equilibria exist. However, the proof in their paper
is not complete. Moreover, the PoS of the existing approximate equilibria is not
discussed. In this section, not only we provide a simpler proof of this result, but
we also extend it for a range of $\lambda$-approximate equilibria with
$\lambda\in[3/2,2]$, and for a guarantee on the PoS.

\begin{lemma}
\label{lemma:concave-good}
Any nondecreasing concave function is $(\mu,\mu+\frac{1}{2},\frac{1}{2},\mu+\frac{1}{2})$-good, for all $\mu\in[\frac{1}{2},1]$.
\end{lemma}
\begin{proof}
Fix a nondecreasing concave function $c:\R_+\map \R_+$ and a parameter $0\leq \xi \leq \frac{1}{2}$.
First note that, since $c$ is nonnegative and concave, it must be subadditive. That is, for all $x,z\geq 0$:
\begin{equation}
\label{eq:subadditive}
c(x)+c(z) \geq c(x+z)
\end{equation}
Furthermore, from the Hermite-Hadamard inequality (see, e.g., \cite{Mitrinovic1985}) and the fact that $c$ is nondecreasing, for any $0\leq a < b$:
\begin{equation}
\label{eq:HH}
\frac{f(a)+f(b)}{2}\leq \frac{1}{b-a}\int_a^b f(t)\, dt\leq f(b).
\end{equation}

Applying first \eqref{eq:HH} for $a=x$ and $b=x+w$ we get that
$$
\frac{c(x)+c(x+w)}{2}\leq \frac{1}{w}\int_x^{x+w} c(t)\, dt \leq c(x+w),
$$
so
\begin{align*}
\frac{1}{w} \int_x^{x+w} c(t)\,dt +\xi\cdot c(w) 
	&\leq c(x+w)+ \xi c(w)\leq (1+\xi)c(x+z)\\
\intertext{and}
\frac{1}{w} \int_x^{x+w} c(t)\,dt +\xi\cdot c(w) 
	& \geq \frac{c(x)+c(x+w)}{2} +\xi c(w)\\
	& \geq \frac{1}{2} c(x+w) + \xi [c(x)+c(w)], &&\text{since}\;\; \xi\leq \frac{1}{2},\\
		& \geq \left(\frac{1}{2}+\xi\right)c(x+w), &&\text{due to \eqref{eq:subadditive}.}
\end{align*}
Thus, condition~\eqref{eq:good_1} is satisfied with $\alpha_{1}=\frac{1}{2}+\xi$, $\alpha_{2}=1+\xi$.

Next, applying~\eqref{eq:HH} for $a=0$ and $b=x$ we get that
\begin{align*}
\frac{1}{2}c(x) \leq \frac{c(0)+c(x)}{2} 
\leq \frac{1}{x} \int_0^x c(t)\,dt 
\leq c(x)=(1+\xi)c(x)-\xi c(x),
\end{align*}
thus condition~\eqref{eq:good_2_increasing} is satisfied with $\beta_{1}=\frac{1}{2}$, $\beta_{2}=1+\xi$.

Summarizing, we have shown (see~\cref{def:good-costs}) that any concave cost
function is $(\xi+\frac{1}{2},\xi+1,\frac{1}{2},\xi+1)$-good, for any
$\xi\in[0,\frac{1}{2}]$. Performing the change of variables $\mu=\xi+\frac{1}{2}$
concludes our proof.
\end{proof}

\begin{theorem}
\label{th:games-concave}
Any weighted congestion game with nondecreasing concave cost functions has a $(\lambda,\frac{\lambda}{\lambda-1})$-equilibrium, for any $\lambda\in[\frac{3}{2},2]$. 
\end{theorem}
\begin{proof}
Fix a parameter $\lambda\in[\frac{3}{2},2]$ and let $\mu=\frac{1}{2(\lambda-1)}$. Then, $\mu\in[\frac{1}{2},1]$ and thus, due to~\cref{lemma:concave-good}, we can deduce that our game is $(\mu,\mu+\frac{1}{2},\frac{1}{2},\mu+\frac{1}{2})$-good (according
to~\cref{def:good-costs}). Deploying~\cref{th:good_existence_easy} we can establish the existence of an
a $(\alpha,\beta)$-equilibrium with 
$$
\alpha=\frac{\mu+\frac{1}{2}}{\mu}=1+\frac{1}{2\mu}=1+(\lambda-1)=\lambda
\qquad\text{and}\qquad
\beta=\frac{(\mu+\frac{1}{2})/\mu}{\frac{1}{2}/\mu}=2\mu+1=\frac{1}{\lambda-1}+1=\frac{\lambda}{\lambda-1}.
$$
\end{proof}

\subsection{Fair Cost Sharing}
\label{sec:fair-costs-simple}

In this section, we focus on the fair cost sharing model in which
$c_e(x)=\frac{a_e}{x}$, where $a_e$ is a positive, resource-dependent value. We
assume that $w_{\min}=1$; this is without loss, since we can just rescale the player
weights. This setting was studied by Chen and Roughgarden~\cite{Chen2008}. Here we
improve on their results (see~\cref{fig:fair-cost}), with a simpler proof.

We must notice that the function $x\mapsto a_e/x$ is not integrable in an interval
starting at $0$, and hence we cannot immediately apply our \cref{def:good-costs}.
However, based on our discussion in~\cref{sec:equivalent-costs} we can modify the
game in order to overcome this. First, we assume for our analysis that $a_e=1$ since
any other choice of $a_e$ can be seen as a trivial conical combination of the
function $1/x$ (see~\cref{def:good-games}). Next, we change the cost function
$c_e(x)$ to be constant and equal to $\lambda$ in the interval $[0,1)$, for some
$\lambda \geq 1$.

\begin{lemma}
\label{lem:fair-costs-good}
Fix a weighted congestion game with $w_{\min}=1$. For any $\lambda\geq 1$, the cost
function

\[c(x)=\begin{cases}1/x, & x\geq 1,\\
\lambda, & 0\leq x<1\end{cases}\]
is $(\alpha_1,\alpha_2,\beta_1,\beta_2)$-good with
\begin{align*}
\alpha_1 &=1, & \alpha_2 &=\max\left(\left(1+\frac{1}{w_{\max}}\right)\ln(1+w_{\max}),\ln(w_{\max})+\lambda\right),\\
\beta_1 &=\lambda, & \beta_2 &=\ln W+\lambda.
\end{align*}
\end{lemma}

\begin{proof}
We will choose $\xi=0$ in the \cref{def:good-costs} of good cost functions. Thus, we need to find nonnegative quantities $\alpha_1$, $\alpha_2$, $\beta_1$, $\beta_2$ such that, for $x\in\{0\}\cup[1,W]$, $w\in[1,w_{\max}]$,

\[\alpha_1\cdot c(x+w)\leq\frac{1}{w}\int_x^{x+w}c(t)dt\leq\alpha_2\cdot c(x+w),\]
and for all $x\in[1,W]$,
\[\beta_1\cdot c(x)\leq\frac{1}{x}\int_0^xc (t)dt\leq \beta_2\cdot c(x).\]

For the bounds in $\alpha_1$, $\alpha_2$, we are interested in upper and lower bounds on the ratio
\[R(w,x)=\frac{\frac{1}{w}\int_x^{x+w}c(t)dt}{c(x+w)}.\]
When $x\geq 1$, this becomes
\[R(w,x)=\frac{\frac{1}{w}(\ln(x+w)-\ln(x))}{\frac{1}{x+w}}=\left(1+\frac{x}{w}\right)\ln\left(1+\frac{w}{x}\right);\]
on the other hand, when $x=0$, this becomes
\[R(w,0)=\frac{\frac{1}{w}(\ln(w)+\lambda)}{\frac{1}{w}}=\ln(w)+\lambda.\]
Thus, we get
\[R(w,x)=\left\{\begin{array}{cc}\left(1+\frac{x}{w}\right)\ln\left(1+\frac{w}{x}\right),&x\geq1;\\ \ln w+\lambda,&x=0.\end{array}\right.\]

In \cref{lem:rwx_fair_cost} in the Appendix, we show that the upper branch of $R(w,x)$ is increasing in $w$ and decreasing in $x$; hence, it is maximized at $w\rightarrow w_{\max}$, $x\rightarrow 1$, for a value of $R(w_{\max},1)=\left(1+\frac{1}{w_{\max}}\right)\ln(1+w_{\max})$; and minimized at $w\rightarrow 1$, $x\rightarrow W$, for a value of $R(1,W)\geq R(1,\infty)=1$. On the other hand, the lower branch is maximized at $w\rightarrow w_{\max}$, for a value of $R(w_{\max},0)=\ln(w_{\max})+\lambda$; and minimized at  $w\rightarrow 1$, for a value of $R(1,0)=\lambda\geq 1$. This gives the desired bounds on $\alpha_1$ and $\alpha_2$.

Next, we look at the bounds in $\beta_1,\beta_2$. Since $x\in[1,W]$ we have that $\int_0^x c(t)dt=\ln x+\lambda$. Moreover, it is immediate to observe that

\[\lambda\cdot c(x)=\lambda\cdot\frac{1}{x}\leq\frac{1}{x}(\ln x+\lambda)\leq\frac{1}{x}(\ln W+\lambda)=(\ln W+\lambda)\cdot c(x).\]
\end{proof}

\begin{theorem}
\label{thm:fair-cost-pne}
Fix a fair cost sharing game with unit minimum weight ($w_{\min}=1$), and let
$w_{\max},W$ be the maximum weight and the maximum total load. Then, for any
$\lambda \geq 1$, our game has an $(\alpha,\beta)$-equilibrium where
\[\alpha=\max\left(\left(1+\frac{1}{w_{\max}}\right)\ln(1+w_{\max}),\ln(w_{\max})+\lambda\right),\qquad\beta=1+\frac{\ln W}{\lambda}.\]
\end{theorem}
\begin{figure}[t]
\centering
\begin{tikzpicture}[baseline=height,scale = 0.6]
\tikzmath{\wm = 10; \a = 10.0; \b= 6.0; \bone=1.386; \btwo=2.443;}

\draw[->] (1, 1) -- (\a, 1);
\node[right] at (\a, 1) {\scriptsize $w_{\max}$};
\draw[->] (1, 1) -- (1, \b);
\node[above] at (1, \b) {\scriptsize $\alpha$};

\node[left] at (1, 1) {\scriptsize $1$};
\node[left] at (1, \bone) {\scriptsize $1.386$};
\node[left] at (1, \btwo) {\scriptsize $2.443$};

\draw[blue,thick,domain=1:\wm,samples=100] plot (\x,{max((1+1/\x)*ln(1+\x),1+ln(\x))});
\draw[red,thick,domain=1:\wm,samples=100] plot (\x,{ln(exp(1)*(1+\x))/ln(2)});

\end{tikzpicture}
\begin{tikzpicture}[baseline=height,scale = 0.6]
\tikzmath{\wm = 3; \wt = 50; \xim=10; \a = 1+\xim + ln(\wm); \b= 1.5 + ln(\wt); \cra=2*ln(exp(1)*(1+\wm))/ln(2); \crb=1+2*ln(1+\wt)/(ln(2)*\cra);}

\draw[->] (1, 1) -- (\a, 1);
\node[right] at (\a, 1) {\scriptsize $\alpha$};
\draw[->] (1, 1) -- (1, \b);
\node[above] at (1, \b) {\scriptsize $\beta$};

\node[left] at (1, 1) {\scriptsize $1$};

\node[below] at ({1+ln(\wm)},1) {\scriptsize $\varTheta(\ln w_{\max})$};
\node[left] at (1,{1+ln(\wt)}) {\scriptsize $\varTheta(\ln W)$};
\draw[dotted] ({1+ln(\wm)},1) -- ({1+ln(\wm)},{1+ln(\wt)}) -- (1,{1+ln(\wt)});

\node[below] at ({\xim + ln(\wm)}, 1) {\scriptsize $\varTheta(\ln W)$};
\node[left] at (1, {1 + ln(\wt)/\xim}) {\scriptsize $\varTheta(1)$};
\draw[dotted] ({\xim + ln(\wm)}, 1) -- ({\xim + ln(\wm)}, {1 + ln(\wt)/\xim}) -- (1, {1 + ln(\wt)/\xim});
\draw[dotted] (\cra, 1) -- (\cra,\crb) -- (1, \crb);

\draw[blue,thick,domain=1:\a-ln(\wm),samples=100] plot ({\x+ln(\wm)},{1+ln(\wt)/\x});
\draw[blue,thick,fill=blue] ({1+ln(\wm)},{1+ln(\wt)}) circle (2pt);

\draw[red,thick,domain=\cra:\a,samples=100] plot ({\x},{1+2*ln(1+\wt)/(ln(2)*\x)});
\draw[red,thick,fill=red] (\cra,\crb) circle (2pt);

\end{tikzpicture}
\caption{Fair cost sharing games. Left: guarantee on the existence of
$\alpha$-approximate equilibria, as a function of $w_{\max}$, given
by~\cref{thm:fair-cost-pne} (setting $\lambda=1$). Right: trade-off curve for the
existence of $(\alpha,\beta)$-equilibria, given by~\cref{thm:fair-cost-pne}; here we
choose $w_{\max}=3$, $W=50$. For comparison, the previously best bounds
\cite[Theorem 5.1 and Lemma 5.3]{Chen2008} are plotted in red, while our results are
in blue. The fact that the blue line of the right plot starts earlier is a direct
consequence of our results providing a strictly better (smaller) absolute existence
guarantee $\alpha$ (see left plot).}
\label{fig:fair-cost}
\end{figure}
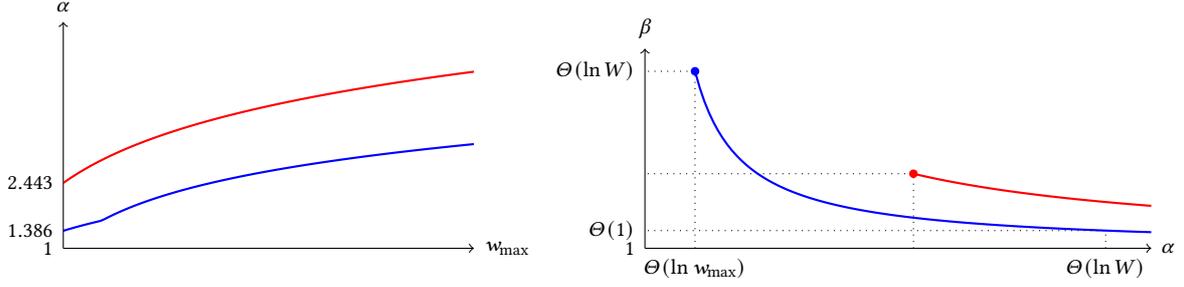

\begin{proof}
Combining \cref{th:good_existence_easy} with \cref{lem:fair-costs-good} we conclude that, for $\lambda\geq1$, our game has an $(\alpha,\beta)$-equilibrium with
\[\alpha=\frac{\alpha_2}{\alpha_1}=\max\left(\left(1+\frac{1}{w_{\max}}\right)\ln(1+w_{\max}),\ln(w_{\max})+\lambda\right),\qquad\beta=\frac{\beta_2/\alpha_1}{\beta_1/\alpha_1}=\frac{\ln W+\lambda}{\lambda}=1+\frac{\ln W}{\lambda}.\]
\end{proof}

The parameter $\lambda$ quantifies the trade-off curve between the approximation guarantee on equilibria and their price of stability. At one extreme case $\lambda=1$, we get that
\[\alpha=\max\left(\left(1+\frac{1}{w_{\max}}\right)\ln(1+w_{\max}),\ln(w_{\max})+1\right)=\varTheta(\ln w_{\max}),\qquad \beta=1+\ln W;\]
in other words, there exist $\varTheta(\ln w_{\max})$-approximate equilibria with price of stability $\varTheta(\ln W)$. At the other extreme case $\lambda=\varTheta(\ln W)$, we get that
\[\alpha=\max\left(\left(1+\frac{1}{w_{\max}}\right)\ln(1+w_{\max}),\ln(w_{\max})+\varTheta(\ln W)\right)=\varTheta(\ln W),\qquad\beta=1+\frac{\ln W}{\varTheta(\ln W)}=\varTheta(1);\]
in other words, there exist $\varTheta(\ln W)$-approximate equilibria with
\emph{constant} price of stability $\varTheta(1)$. The complete trade-off
curve can be seen in~\cref{fig:fair-cost} (right). We can also compare our results with the best known upper bounds. In \cite[Lemma 5.3]{Chen2008}, it was shown that $\alpha$-approximate equilibria exist for $\alpha\geq\log_2[e(1+w_{\max})]$; and in \cite[Theorem 5.1]{Chen2008}, it was shown that $\left(f,1+\frac{2\log_2(1+W)}{f}\right)$-equilibria exist for any $f\geq
2\log_2[e(1+w_{\max})]$. As \cref{fig:fair-cost} shows, we improve on both results.

\subsection{Mixtures of Cost Functions}
\label{sec:costs-mixtures}
A big advantage of our approach is that we can study the existence of
$(\alpha,\beta)$-equilibria for games that merge cost functions of two or more
different types. For example, in this section we look at congestion games that have
\emph{both} concave costs and polynomial costs (as well as any conical combination).
Interestingly, we show that this results in only a small increase in the PoS
guarantee of~\cref{th:games-poly}, while the existence guarantee stays the same. For
the following theorem we consider polynomials of degree at least $2$, since affine
functions are themselves concave and would be already captured
by~\cref{th:games-concave}.

\begin{theorem}
\label{th:games-concave-poly}
Any weighted congestion game with cost functions that are conical combinations of
concave and polynomial costs of maximum degree $d\geq 2$ has an
$(\lambda,1+\frac{d+1}{\lambda})$-equilibrium, for any $\lambda\in[d,d+1]$.
\end{theorem}
\begin{proof}
Fix a maximum degree $d\geq 2$ for the polynomial costs and a parameter
$\lambda\in[d,d+1]$. By defining $\mu=\frac{d+1}{2\lambda}$ we have that
$\frac{1}{2}\leq \mu \leq \frac{1}{2}\left(1+\frac{1}{d}\right) \leq 1$, and so by
applying~\cref{lemma:concave-good} we can derive that any concave cost is
$(\mu,\mu+\frac{1}{2},\frac{1}{2},\mu+\frac{1}{2})$-good. Next, by~\cref{lemma:polynomials-good,lemma:constant-good} we can derive that all
monomials of degree $k=0,\dots,d-1$ are
$(\frac{1}{k+1},1,\frac{1}{k+1},\frac{1}{k+1})$-good and the monomial of degree $d$
is $(\frac{1}{\lambda},1,\frac{1}{d+1},\frac{1}{\lambda})$-good.

Deploying our black-box tool~\cref{th:good_existence_easy} (and shortcutting some calculations that we have already performed in the proof of~\cref{th:games-poly}) we can guarantee the existence of an $(\alpha,\beta)$-equilibrium with
\begin{equation*}
\alpha 
	= \max\sset{1+\frac{1}{2\mu},\lambda}
	= \max\sset{1+\frac{\lambda}{d+1},\lambda}
	= \lambda,
\end{equation*}
since $2\leq d\leq \lambda \leq d+1$, and
\begin{equation*}
\beta 
	= \frac{\max\sset{\frac{\mu+1/2}{\mu},1}}{\min\sset{\frac{1/2}{\mu},\frac{\lambda}{d+1}}}
	= \frac{1+\frac{1}{2\mu}}{\min\sset{\frac{1}{2\mu},\frac{\lambda}{d+1}}}
	= 1+2\mu = 1+ \frac{d+1}{\lambda},
\end{equation*}
where for the third equality we used that, from the definition of $\mu$, $\frac{1}{2\mu}= \frac{\lambda}{d+1}$.
\end{proof}

\subsubsection*{Acknowledgements}
We thank Martin Gairing for interesting discussions.

\bibliographystyle{alphaurl} 
\bibliography{unifying_potential}

\newcommand{\etalchar}[1]{$^{#1}$}
\begin{thebibliography}{CGG{\etalchar{+}}20}

\bibitem[ADK{\etalchar{+}}08]{Anshelevich2008a}
Elliot Anshelevich, Anirban Dasgupta, Jon Kleinberg, {\'{E}}va Tardos, Tom
  Wexler, and Tim Roughgarden.
\newblock The price of stability for network design with fair cost allocation.
\newblock {\em {SIAM} Journal on Computing}, 38(4):1602--1623, January 2008.
\newblock \href {https://doi.org/10.1137/070680096}
  {\path{doi:10.1137/070680096}}.

\bibitem[BCFM13]{Bilo2013}
Vittorio Bil{\`{o}}, Ioannis Caragiannis, Angelo Fanelli, and Gianpiero Monaco.
\newblock Improved lower bounds on the price of stability of undirected network
  design games.
\newblock {\em Theory of Computing Systems}, 52(4):668--686, May 2013.
\newblock \href {https://doi.org/10.1007/s00224-012-9411-6}
  {\path{doi:10.1007/s00224-012-9411-6}}.

\bibitem[BFM14]{Bilo2014}
Vittorio Bil{\`o}, Michele Flammini, and Luca Moscardelli.
\newblock The price of stability for undirected broadcast network design with
  fair cost allocation is constant.
\newblock {\em Games and Economic Behavior}, 2014.
\newblock \href {https://doi.org/10.1016/j.geb.2014.09.010}
  {\path{doi:10.1016/j.geb.2014.09.010}}.

\bibitem[CF19]{Caragiannis:2019aa}
Ioannis Caragiannis and Angelo Fanelli.
\newblock On approximate pure {Nash} equilibria in weighted congestion games
  with polynomial latencies.
\newblock In {\em Proceedings of the 46th International Colloquium on Automata,
  Languages, and Programming (ICALP)}, pages 133:1--133:12, 2019.
\newblock \href {https://doi.org/10.4230/LIPIcs.ICALP.2019.133}
  {\path{doi:10.4230/LIPIcs.ICALP.2019.133}}.

\bibitem[CFGS11]{Caragiannis2011}
Ioannis Caragiannis, Angelo Fanelli, Nick Gravin, and Alexander Skopalik.
\newblock Efficient computation of approximate pure {Nash} equilibria in
  congestion games.
\newblock In {\em Proceedings of the 52nd IEEE Annual Symposium on Foundations
  of Computer Science (FOCS)}, pages 532--541, 2011.
\newblock \href {https://doi.org/10.1109/focs.2011.50}
  {\path{doi:10.1109/focs.2011.50}}.

\bibitem[CGG{\etalchar{+}}20]{cggpw2020}
George Christodoulou, Martin Gairing, Yiannis Giannakopoulos, Diogo Poças, and
  Clara Waldmann.
\newblock Existence and complexity of approximate equilibria in weighted
  congestion games.
\newblock In {\em Proceedings of the 47th International Colloquium on Automata,
  Languages, and Programming (ICALP)}, pages 32:1--32:18, 2020.
\newblock \href {https://doi.org/10.4230/LIPIcs.ICALP.2020.32}
  {\path{doi:10.4230/LIPIcs.ICALP.2020.32}}.

\bibitem[CGGS19]{cggs2018}
George Christodoulou, Martin Gairing, Yiannis Giannakopoulos, and Paul~G.
  Spirakis.
\newblock The price of stability of weighted congestion games.
\newblock {\em SIAM Journal on Computing}, 48(5):1544--1582, 2019.
\newblock \href {https://doi.org/10.1137/18M1207880}
  {\path{doi:10.1137/18M1207880}}.

\bibitem[CKS11]{Christodoulou2011a}
George Christodoulou, Elias Koutsoupias, and Paul~G. Spirakis.
\newblock On the performance of approximate equilibria in congestion games.
\newblock {\em Algorithmica}, 61(1):116--140, 2011.
\newblock \href {https://doi.org/10.1007/s00453-010-9449-2}
  {\path{doi:10.1007/s00453-010-9449-2}}.

\bibitem[CR09]{Chen2008}
Ho-Lin Chen and Tim Roughgarden.
\newblock Network design with weighted players.
\newblock {\em Theory of Computing Systems}, 45(2):302--324, July 2009.
\newblock \href {https://doi.org/10.1007/s00224-008-9128-8}
  {\path{doi:10.1007/s00224-008-9128-8}}.

\bibitem[CSSM04]{Correa2004}
Jos{\'{e}}~R. Correa, Andreas~S. Schulz, and Nicol{\'{a}}s~E. Stier-Moses.
\newblock Selfish routing in capacitated networks.
\newblock {\em Mathematics of Operations Research}, 29(4):961--976, 2004.
\newblock \href {https://doi.org/10.1287/moor.1040.0098}
  {\path{doi:10.1287/moor.1040.0098}}.

\bibitem[DS08]{Dunkel2008}
Juliane Dunkel and Andreas~S. Schulz.
\newblock On the complexity of pure-strategy {Nash} equilibria in congestion
  and local-effect games.
\newblock {\em Mathematics of Operations Research}, 33(4):851--868, 2008.
\newblock \href {https://doi.org/10.1287/moor.1080.0322}
  {\path{doi:10.1287/moor.1080.0322}}.

\bibitem[FHP16]{Freeman2016}
Rupert Freeman, Samuel Haney, and Debmalya Panigrahi.
\newblock {\em On the Price of Stability of Undirected Multicast Games}, pages
  354--368.
\newblock 2016.
\newblock \href {https://doi.org/10.1007/978-3-662-54110-4_25}
  {\path{doi:10.1007/978-3-662-54110-4_25}}.

\bibitem[FKK{\etalchar{+}}09]{Fotakis2009a}
Dimitris Fotakis, Spyros Kontogiannis, Elias Koutsoupias, Marios Mavronicolas,
  and Paul Spirakis.
\newblock The structure and complexity of {Nash} equilibria for a selfish
  routing game.
\newblock {\em Theoretical Computer Science}, 410(36):3305--3326, 2009.
\newblock \href {https://doi.org/10.1016/j.tcs.2008.01.004}
  {\path{doi:10.1016/j.tcs.2008.01.004}}.

\bibitem[FKL{\etalchar{+}}06]{Fiat2006}
Amos Fiat, Haim Kaplan, Meital Levy, Svetlana Olonetsky, and Ronen Shabo.
\newblock On the price of stability for designing undirected networks with fair
  cost allocations.
\newblock In {\em Proceedings of the 33rd International ColloquiumAutomata,
  Languages and Programming (ICALP)}, pages 608--618, 2006.
\newblock \href {https://doi.org/10.1007/11786986_53}
  {\path{doi:10.1007/11786986_53}}.

\bibitem[FKS05]{Fotakis2005a}
Dimitris Fotakis, Spyros Kontogiannis, and Paul Spirakis.
\newblock Selfish unsplittable flows.
\newblock {\em Theoretical Computer Science}, 348(2):226--239, 2005.
\newblock \href {https://doi.org/10.1016/j.tcs.2005.09.024}
  {\path{doi:10.1016/j.tcs.2005.09.024}}.

\bibitem[GMV05]{Goemans2005}
M.~Goemans, Vahab Mirrokni, and A.~Vetta.
\newblock Sink equilibria and convergence.
\newblock In {\em Proceedings of the 46th Annual IEEE Symposium on Foundations
  of Computer Science (FOCS)}, pages 142--151, 2005.
\newblock \href {https://doi.org/10.1109/SFCS.2005.68}
  {\path{doi:10.1109/SFCS.2005.68}}.

\bibitem[GP20]{gp2020_sagt}
Yiannis Giannakopoulos and Diogo Poças.
\newblock A unifying approximate potential for weighted congestion games.
\newblock In {\em Proceedings of the 13th Symposium on Algorithmic Game Theory
  (SAGT)}, pages 99--113, 2020.
\newblock \href {https://doi.org/10.1007/978-3-030-57980-7_7}
  {\path{doi:10.1007/978-3-030-57980-7_7}}.

\bibitem[HK12]{Harks2012a}
Tobias Harks and Max Klimm.
\newblock On the existence of pure {Nash} equilibria in weighted congestion
  games.
\newblock {\em Mathematics of Operations Research}, 37(3):419--436, 2012.
\newblock \href {https://doi.org/10.1287/moor.1120.0543}
  {\path{doi:10.1287/moor.1120.0543}}.

\bibitem[HKM12]{Harks2012}
Tobias Harks, Max Klimm, and Rolf~H M{\"{o}}hring.
\newblock Strong equilibria in games with the lexicographical improvement
  property.
\newblock {\em International Journal of Game Theory}, 42(2):461--482, 2012.
\newblock \href {https://doi.org/10.1007/s00182-012-0322-1}
  {\path{doi:10.1007/s00182-012-0322-1}}.

\bibitem[HKS14]{Hansknecht2014}
Christoph Hansknecht, Max Klimm, and Alexander Skopalik.
\newblock Approximate pure {Nash} equilibria in weighted congestion games.
\newblock In {\em Proceedings of the 17th International Workshop on
  Approximation Algorithms for Combinatorial Optimization Problems (APPROX)},
  pages 242--257, 2014.
\newblock \href {https://doi.org/10.4230/LIPIcs.APPROX-RANDOM.2014.242}
  {\path{doi:10.4230/LIPIcs.APPROX-RANDOM.2014.242}}.

\bibitem[KP99]{Koutsoupias1999b}
Elias Koutsoupias and Christos Papadimitriou.
\newblock Worst-case equilibria.
\newblock In {\em Proceedings of the 16th Annual Symposium on Theoretical
  Aspects of Computer Science (STACS)}, pages 404--413, 1999.
\newblock \href {https://doi.org/10.1016/j.cosrev.2009.04.003}
  {\path{doi:10.1016/j.cosrev.2009.04.003}}.

\bibitem[LO01]{Libman2001}
Lavy Libman and Ariel Orda.
\newblock Atomic resource sharing in noncooperative networks.
\newblock {\em Telecommunication Systems}, 17(4):385--409, August 2001.
\newblock \href {https://doi.org/10.1023/A:1016770831869}
  {\path{doi:10.1023/A:1016770831869}}.

\bibitem[ML85]{Mitrinovic1985}
D.~S. Mitrinović and I.~B. Lacković.
\newblock Hermite and convexity.
\newblock {\em Aequationes mathematicae}, 28:229--232, 1985.
\newblock URL: \url{http://eudml.org/doc/137060}.

\bibitem[MS96]{Monderer1996a}
Dov Monderer and Lloyd~S. Shapley.
\newblock Potential games.
\newblock {\em Games and Economic Behavior}, 14(1):124--143, 1996.
\newblock \href {https://doi.org/10.1006/game.1996.0044}
  {\path{doi:10.1006/game.1996.0044}}.

\bibitem[NRTV07]{2007a}
Noam Nisan, Tim Roughgarden, {\'{E}}va Tardos, and Vijay Vazirani, editors.
\newblock {\em Algorithmic Game Theory}.
\newblock Cambridge University Press, 2007.
\newblock \href {https://doi.org/10.1017/CBO9780511800481}
  {\path{doi:10.1017/CBO9780511800481}}.

\bibitem[Pap01]{Papadimitriou2001a}
Christos Papadimitriou.
\newblock Algorithms, games, and the internet.
\newblock In {\em Proceedings of the 33rd Annual ACM Symposium on Theory of
  Computing (STOC)}, pages 749--753, 2001.
\newblock \href {https://doi.org/10.1145/380752.380883}
  {\path{doi:10.1145/380752.380883}}.

\bibitem[PS07]{Panagopoulou2007}
Panagiota~N. Panagopoulou and Paul~G. Spirakis.
\newblock Algorithms for pure {Nash} equilibria in weighted congestion games.
\newblock {\em Journal of Experimental Algorithmics}, 11:27, February 2007.
\newblock \href {https://doi.org/10.1145/1187436.1216584}
  {\path{doi:10.1145/1187436.1216584}}.

\bibitem[Ros73]{Rosenthal1973a}
Robert~W. Rosenthal.
\newblock A class of games possessing pure-strategy {Nash} equilibria.
\newblock {\em International Journal of Game Theory}, 2(1):65--67, 1973.
\newblock \href {https://doi.org/10.1007/BF01737559}
  {\path{doi:10.1007/BF01737559}}.

\bibitem[Rou07]{Roughgarden2007}
Tim Roughgarden.
\newblock Routing games.
\newblock In Noam Nisan, Tim Roughgarden, {\'{E}}va Tardos, and Vijay Vazirani,
  editors, {\em Algorithmic Game Theory}, chapter~18. Cambridge University
  Press, 2007.
\newblock \href {https://doi.org/10.1017/CBO9780511800481.020}
  {\path{doi:10.1017/CBO9780511800481.020}}.

\bibitem[Rou16]{Roughgarden2016}
Tim Roughgarden.
\newblock {\em Twenty Lectures on Algorithmic Game Theory}.
\newblock Cambridge University Press, 2016.
\newblock \href {https://doi.org/10.1017/cbo9781316779309}
  {\path{doi:10.1017/cbo9781316779309}}.

\bibitem[TW07]{Tardos:2007aa}
{\'E}va Tardos and Tom Wexler.
\newblock Network formation games and the potential function method.
\newblock In Noam Nisan, Tim Roughgarden, {\'{E}}va Tardos, and Vijay Vazirani,
  editors, {\em Algorithmic Game Theory}, chapter~19. Cambridge University
  Press, 2007.
\newblock \href {https://doi.org/10.1017/cbo9780511800481.021}
  {\path{doi:10.1017/cbo9780511800481.021}}.

\bibitem[V{\"{o}}c07]{Vocking2007a}
Berthold V{\"{o}}cking.
\newblock Selfish load balancing.
\newblock In Noam Nisan, Tim Roughgarden, {\'{E}}va Tardos, and Vijay Vazirani,
  editors, {\em Algorithmic Game Theory}, chapter~20. Cambridge University
  Press, 2007.
\newblock \href {https://doi.org/10.1017/cbo9780511800481.022}
  {\path{doi:10.1017/cbo9780511800481.022}}.

\end{thebibliography}

\appendix

\section{Technical Lemmas}

\begin{lemma}\label{lem:rwx_fair_cost}
For $w,x\in[1,\infty)$, the function
\[R(w,x)=\left(1+\frac{x}{w}\right)\ln\left(1+\frac{w}{x}\right)\]
is increasing in $w$ and decreasing in $x$.
\end{lemma}

\begin{proof} Let us apply the change of variables $\frac{1}{z}\equiv 1+\frac{x}{w}$, so that we can write
\[R(w,x)=\left(1+\frac{x}{w}\right)\ln\left(1+\frac{w}{x}\right)\equiv\frac{\ln\left(\frac{1}{1-z}\right)}{z}.\]
Since $w,x\in[1,\infty)$, it follows that $1/z\in(1,\infty)$ and thus $z\in(0,1)$. Notice now that $\ln\left(\frac{1}{1-z}\right)$ is convex for $z\in(0,1)$, since its first derivative,
\[\frac{d}{dz}\ln\left(\frac{1}{1-z}\right)=-\frac{d}{dz}\ln\left(1-z\right)=\frac{1}{1-z}\]
is increasing in $z$. Since in addition $\left.\ln\left(\frac{1}{1-z}\right)\right|_{z=0}=0$, we conclude that $\ln\left(\frac{1}{1-z}\right)/z$ is increasing in $z$. Since $z=\frac{w}{w+x}$ is increasing in $w$ and decreasing in $x$, the result follows.
\end{proof}

\end{document}